\documentclass[journal,twocolumn]{IEEEtran}

\usepackage{amsfonts, amsthm}
\usepackage{amsmath, color}
\usepackage{graphicx}
\usepackage{cite}
\usepackage[normalem]{ulem}
\usepackage{tikz}
\usepackage{epstopdf}
\usepackage{psfrag}
\usepackage{etoolbox}
\usepackage[process=auto]{pstool}
\usetikzlibrary{arrows,automata}
\usepackage{graphicx}

\newtheorem{definition}{Definition}

\newtheorem{problem}{Problem}

\newtheorem{proposition}{Proposition}

\newtheorem{example}{Example}
\newtheorem{property}{Property}
\allowdisplaybreaks
\title{A Bio-Synthetic Modulator Model for Diffusion-based Molecular Communications}
\author{Hamidreza~Arjmandi$^*$, Arman~Ahmadzadeh$^{\dagger}$,  Robert~Schober$^{\dagger}$, Masoumeh~Nasiri~Kenari$^*$\\
$^*$ Sharif~University~of~Technology, $^{\dagger}$ Friedrich-Alexander~University~of~Erlangen-Nuremberg
\thanks{This paper is an extended version of a paper submitted to IEEE Globecom
2016.}
\vspace{-0.5cm}}

\begin{document}
\maketitle

\begin{abstract}
In diffusion-based molecular communication (DMC), one important functionality of a transmitter nano-machine is signal modulation. In particular, the transmitter has to be able to control the release of signaling molecules for modulation of the information bits. 
An important class of control mechanisms in natural cells for releasing molecules is based on ion channels which are pore-forming proteins across the cell membrane whose opening and closing may be controlled by a gating parameter. 
In this paper, a modulator for DMC based on ion channels is proposed which controls the rate at which molecules are released from the transmitter by modulating a gating parameter signal. Exploiting the capabilities of the proposed modulator, an on-off keying modulation scheme is introduced and the corresponding average modulated signal, i.e., the average release rate of the molecules from the transmitter, is derived in the Laplace domain. By making a simplifying assumption, a closed-form expression for the average modulated signal in the time domain is obtained which constitutes an upper bound on the total number of released molecules regardless of this assumption. 
The derived average modulated signal is compared to results obtained with a particle based simulator. The numerical results show that the derived upper bound is tight if the number of ion channels distributed across the transmitter (cell) membrane is small.  
\end{abstract}
\vspace{-0.3cm}
\section{Introduction}
Diffusion-based molecular communication (DMC) is a promising approach for communication among nano-machines. Thereby, the information bits are encoded (or modulated) in the concentration, type, or release time of the diffusing molecules \cite{Akyl2011}. DMC is prevalent in natural communication systems. Therefore, it is expected that the functionalities of the nano-machines required for synthetic DMC systems can be modeled and designed based on mechanisms already existing in biological systems \cite{NEH13, Farsad2014}. One important functionality required for DMC is signal modulation at the transmitter nano-machine. In particular, the transmitter has to be able to control the release of the signaling molecules for modulation of the information to be transmitted. 
    
In most of the existing molecular communication (MC) literature, the transmitter is modeled as an ideal point source which can release any desired number of molecules instantaneously at the beginning of a symbol interval \cite{R1,MMM14,Nakano2013}. The possibility of ``pulse shaping", i.e., the non-instantaneous but still deterministic release of signaling molecules, was considered in \cite{Garralda11a}. However, in a real system, the transmitter will be a biological or electronic nano-machine (e.g. a modified cell) with a size on the order of tens of nanometer to a few micrometer \cite{NEH13}, which has to generate the signaling molecules via some chemical process and control the release of these signaling molecules into the channel using e.g. electrical, chemical, or optical signals \cite{Alberts14}.
Due to the non-zero time constants and the inherent randomness of the molecule generation and release processes, the molecules will not enter the channel instantaneously and their number will not be deterministic.
In \cite{PA10}, the transmitter is modeled as a box in which the concentration of molecules can be controlled arbitrarily such that a desired concentration gradient between the inside and the outside of the transmitter is obtained. 
In \cite{Arjmandi2013}, the transmitter is assumed to have a storage of molecules and the release of the molecules is adjusted by the outlet size which can be controlled. The authors model the total number of molecules that leave the transmitter by a Poisson distribution. 
Moreover, the author in \cite{Chou2015} assumes the emission patterns of different symbols are generated by different chemical reactions.
However, the models in \cite{Arjmandi2013, PA10, Chou2015} do not include the physical characteristics of the mechanism controlling molecule release such as the transmitter geometry and forces caused by concentration gradients or pumping mechanisms. 

Inspired by nature, a bio-synthetic MC system which mimics the moth pheromone system, has been recently proposed in \cite{Olsson2015}. As part of this system, a micro-machined evaporator was developed to release pheromones. Thereby, temperature regulation controls the timely and precise release of the pheromones. In contrast, the main goal of this paper is to propose a modulator model for DMC by employing the control mechanisms of natural cells.
In cell biology \cite[Chapter 12]{Alberts14}, \cite[Chapter 3]{Blass2015}, different mechanisms for releasing molecules by transporting them across the cell membrane (lipid bilayer) are known. Thereby, two main classes of transport mechanisms maybe distinguished, namely ion channels 
and transporters (also known as "pumps"). Specifically, ion channels are pore-forming membrane proteins that may be voltage gated, ligand gated, mechanically gated, or temperature gated, i.e., the opening and closing of the ion channel may be controlled by an electrical potential, a chemical reaction with a ligand, a mechanical force, or distinct thermal thresholds. Voltage and ligand gating are the two most studied gating mechanisms in the literature \cite{Blass2015}, see Fig. \ref{ionchannelsfig2}. 
\begin{figure}
\begin{center}
\includegraphics[scale=0.2,angle=0]{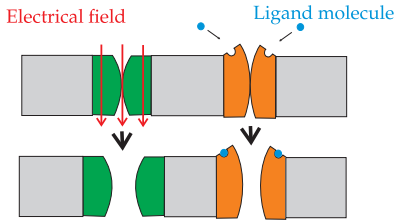}
\end{center}
\vspace{-0.5cm}
\caption{Schematic illustration of voltage and ligand gated ion channels.}
\label{ionchannelsfig2}
\vspace{-0.65cm}
\end{figure}

In this paper, a spherical synthesized cell is considered as transmitter\footnote{The terms "transmitter" and "cell" are used interchangeably in the remainder of the paper.} where the membrane is uniformly covered by voltage gated ion channels. Nevertheless, the obtained results can be easily extended to ligand gated ion channels. Exploiting the ion channels, we propose a modulator for DMC which we refer to as ion channel based bio-synthetic modulator (IBM). The IBM controls the rate at which molecules are released from the cell by modulating the gating parameter signal. Thereby, the release rate of the molecules from the cell constitutes the \textit{modulated signal}.   
To analyze and design IBMs, a simple time slotted on-off keying modulation scheme is considered where the on and off states are controlled by applying different voltage signals. More accurately, to transmit bits 0 and 1 in a given time slot, voltage signals are applied to the cell membrane such that the ion channels close and open, respectively.
The average release rate (the average modulated signal) is analyzed based on the diffusion equations inside and outside the transmitter and the boundary conditions at the transmitter membrane. To this end, an analytical expression for the Laplace transform of the average modulated signal is derived, which generally requires numerical inversion to obtain the time domain signal. The derived average modulated signal is compared to simulation results obtained with a particle based simulator (PBS) \cite{R1}. 
In addition, based on the  simplifying assumption of zero concentration outside the cell, a closed-form expression for the time domain average modulated signal is obtained which leads to an upper bound on the total number of released molecules (time integral of the average modulated signal) regardless of this assumption. Our numerical results show that the derived upper bound is tight if the number of ion channels is small.

\section{Ion channel based Bio-Synthetic Modulator}

  
In nature, voltage gated ion channels play a vital role by transporting inorganic ions across the cell membrane \cite{Alberts14}. 
Thereby, the ion channel proteins open and close randomly. The opening probability is related to conformational changes in the protein which depend on the gating parameter.
The simplest model for voltage gated ion channels is a two-state Markov model having an open and a closed state. Denoting the transition rates from the closed state to the open state (derivative of the transition probability with respect to time) by $\alpha_1(V(t))$ and $\alpha_2(V(t))$, the opening probability of the ion channel, $P_o(t)$, is characterized by the following differential equation \cite{Destex1994}
\begin{align}\label{ODEchannel}
\frac{dP_o(t)}{dt}=\alpha_1(V(t)) (1-P_o(t))-\alpha_2(V(t)) P_o(t),
\end{align}
where $V(t)$ denotes the electrical voltage applied to the voltage gated ion channel. Considering \eqref{ODEchannel}, the opening probability of the ion channel over the cell can be controlled by voltage signal, $V(t)$. On the other hand, the release rate of the molecules from the cell depends on the opening probability of the ion channel. Hence, the release rate of the molecules can be controlled by modulating the voltage signal. 

In the remainder of this section, the IBM transmitter model and a corresponding on-off keying modulation are introduced.

%
 \vspace{-0.3cm}

\subsection{IBM Transmitter Model}
The IBM transmitter is modeled as a spherical cell of radius $r_m$ covered by a membrane, see Fig. \ref{Transmitterfig}. The ion channel proteins are embedded in the cell membrane using biological engineering approaches. 
The specific type of ions that the mounted ion channels allow to pass and are used for signaling are referred to as type $A$ ions (molecules). Assume there are $N$ ion channels uniformly distributed over the cell membrane. 
The radius of the ion channels in the open and closed states is $r_c$ and zero, respectively. The membrane thickness is typically 7-9 nm which is negligible compared to typical cell radii which are on the order of hundreds of nanometer to a few micrometer. Therefore, we divide the space into an intracellular and an extracellar environment which are separated by an infinitely thin membrane. The diffusion constant of the $A$ ions is equal to $D_1$ and $D_2$ inside and outside the cell, respectively.  

For generation of the ions inside the cell, we adopt a simple model as our main focus is the modulation of the ion channels. In particular, we assume that the $A$ ions are generated by a spherical organelle with radius $r_s$, where $r_s \ll r_m$, that is located at the center of the cell. The organelle consumes the energy available inside the cell, produces the signaling molecules by a chemical reaction, and releases the molecules at radius $r=r_s$. A simplified chemical reaction model is considered where the ions are produced at a constant rate of $\mathcal{S}$ molecule (mo)$/$(m$^2$s) and the production is stopped when the average concentration inside the cell reaches a given threshold, $\mathcal{T}$ mo{$/\text{m}^3$}, corresponding to the equilibrium state of the reaction. 
We note that because of their charge, the ions repel each other which results in a drift of the ions towards the boundary of the cell. Our PBS results suggest that this effect is negligible for typical cell sizes, see Section IV. Hence, electrical drift is not taken into account in the analysis presented in Section III.
\begin{figure}
\vspace{-0.2cm}
\begin{center}
\includegraphics[scale=0.32,angle=0]{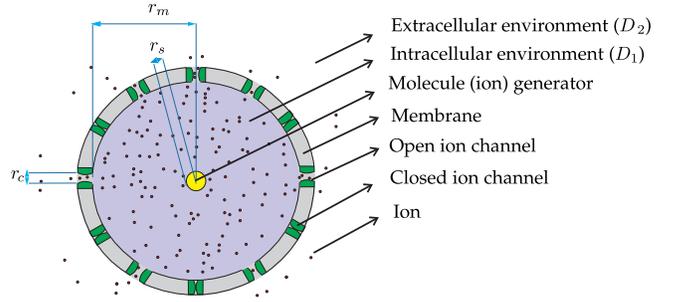}
\end{center}
\vspace{-0.5cm}
\caption{Schematic illustration of IBM transmitter.}
\label{Transmitterfig}
\vspace{-0.5cm}
\end{figure}
\vspace{-0.3cm}
\subsection{IBM Based On-off Keying Modulation}
A simple on-off keying modulation scheme employing the proposed  IBM is adopted in this paper. Assume time is divided into slots of length $T$. Bits 0 and 1 are represented by not releasing and releasing molecules at the transmitter, respectively. To implement this modulation scheme, the ion channel gating voltage, $V(t)$, is utilized to control the release of the produced molecules. For transmission of bit 0, the ion channels should be in the closed state with high probability (close to 1) or equivalently in the open state with small probability (close to 0) during the entire time slot. In contrast, for transmission of bit 1, the ion channels should be in the open state with high probability for duration $T_1$ where $0 < T_1 < T$, and in the closed state during the remainder of the time slot, i.e., $T_2=T-T_1$ seconds.  Equivalently, the desired opening probabilities to transmit bits 0 and 1 in time slot $[0,T]$ are $\mathcal{P}_o^0(t)=0$ and $\mathcal{P}_o^1(t)=u(t)-u(t-T_1)$, respectively, where $u(t)$ is the unit step function. 

We note that molecules are not released during $[T_1,T]$ for transmission of bit 1 for two reasons. First, since the molecules inside the transmitter are released during $[0,T_1]$, the time interval $[T_1,T]$ is needed to replenish the transmitter with molecules in preparation for the next symbol interval. Thereby, we assume $T_2$ is sufficiently large such that the concentration of the molecules inside the cell reaches $\mathcal{T}$ during $[T_1,T]$. Second, not releasing molecules in the diffusion channel during the last $T_2$ seconds of the current symbol interval allows cleansing of the channel from the previous released molecules, if $T_2$ exceeds the length of the diffusion channel memory. Hence, intersymbol interference (ISI) is avoided at the cost of a reduced symbol rate. 

To control the ion channel opening probability, we have to examine how the ion channel opening probability changes with the gating parameter. 
For voltage gated ion channels, the opening probability is given by \eqref{ODEchannel}.      
Assume the initial opening probability at $t=0$ is equal to $P_o^0$ and a constant voltage signal $V_0$ is applied to the membrane starting from $t=0$, i.e., $V(t)=V_{0} u(t)$. Solving differential equation (\ref{ODEchannel}) for this case yields
\vspace{-0.2cm}
 \begin{align}\label{openprob}
P_o(t)=(P_o^{\infty}(V_0)+K(V_0) e^{-t/{t_c(V_0)}})u(t),
\end{align}
where $P_o^{\infty}(V_0)=\frac{\alpha_1(V_{0})}{\alpha_1(V_{0})+\alpha_2(V_{0})}$, $K(V_0)=P_o^0-P_o^{\infty}(V_0)$, and $t_c(V_0)=\frac{1}{\alpha_1(V_{0})+\alpha_2(V_{0})}$. The transition rates $\alpha_1(V_{0})$ and $\alpha_2(V_{0})$ are non-negetive functions of $V_0$ which are obtained by experimental methods, see Example \ref{example1} below. 
From \eqref{openprob}, it is observed that the opening probability approaches the final value of $P_o^{\infty}(V_0)$ exponentially fast with time constant $t_c(V_0)$. Exploiting this simple exponential behavior, we can find voltage values such that the final opening probability approaches zero and one, respectively. In other words, two voltage values denoted by $V_{\rm off}$ and $V_{\rm on}$ can be found such that $P_o^{\infty}(V_{\rm off})\simeq 0$ and $P_o^{\infty}(V_{\rm on})\simeq 1$. Therefore, if $t_c(V_{\rm off})$ and $t_c(V_{\rm on})$ are sufficiently small, applying voltage signals of the form $\mathcal{V}^0(t)=V_{\rm off} (u(t)-u(t-T))$ and $\mathcal{V}^1(t)=V_{\rm on} (u(t)-u(t-T_1))+V_{\rm off}(u(t-T_1)-u(t-T))$ results in opening probabilities close to the desired opening probabilities in time slot $[0,T]$, i.e., $\mathcal{P}_o^0(t)=0$ and $\mathcal{P}_o^1(t)=u(t)-u(t-T_1)$, respectively. 
\begin{example}\label{example1}
Assume a Potassium ion channel with a single gate which has the following transition rates \cite{Hod1952}\footnote{The  model proposed for the Pottassium ion channel in \cite{Hod1952} has several independent gates and each one of them has two states of open and closed. For simplicity, here we assume a Pottassium ion channel with a single gate.};
\vspace{-0.2cm}
\begin{align}
\alpha_1(V(t))&=\frac{0.01 (V(t)+10)}{\exp{\left(\frac{V(t)+10}{10}\right)}-1}\\
\alpha_2(V(t))&=0.125 \exp{\left(\frac{V(t)}{80}\right)},
\end{align}
where the transition rate and the voltage are given in 1$/$ms and mv (millivolts), respectively.

In Fig. \ref{fig0_Opening_prob}, we have plotted the opening probability in response to the applied voltage signal $\mathcal{V}^1(t)$ for different values of $V_{\rm on}=25, 50, 200, -200$ mv and $V_{\rm off}=-200$ mv where $T_1=20$ ms and $T=40$ ms. The adopted value of $V_{\rm off}=-200$ mv, ensures practically zero opening probability for the ion channels whereas the corresponding time constant is $t_c(-200)=0.63$ ms. Numerical inspection of $P_o^{\infty}(V_0)$ shows that it is an increasing function of $V_0$ and for voltages smaller than $200$ mv the opening probability does not approach 1, i.e., only a fraction of the ion channels are open. Also, for voltages smaller than $-50$ mv the opening probability approaches zero. Furthermore, time constant $t_c(V_0)$, which determines the rate at which the opening probability approaches its final value, has a bell shape with a maximum of about $6$ ms for $V_0$ around $-20$ mv and for voltages with magnitude higher than $200$ mv it is smaller than $0.7$ ms. 
The proposed on-off modulation format can be implemented by using voltage signals $\mathcal{V}^1(t)=200 (u(t)-u(t-T_1))-200(u(t-T_1)-u(t-T))$ and $\mathcal{V}^0(t)=-200 (u(t)-u(t-T))$ for bits 1 and 0, respectively. From Fig. \ref{fig0_Opening_prob}, we observe that the resulting opening probability signal is not a perfect rectangle. The deviation arises because of the non-zero time constant, $t_c$, characterizing the exponential behavior of the opening probability. However, this effect is negligible compared to the time constants arising in the molecule release from the transmitter as will be confirmed by numerical results in Section IV. 
\end{example}
\vspace{-0.2cm} 
%
\begin{figure}
\vspace{-0.7cm}
  \centering
\resizebox{1\linewidth}{!}{\psfragfig{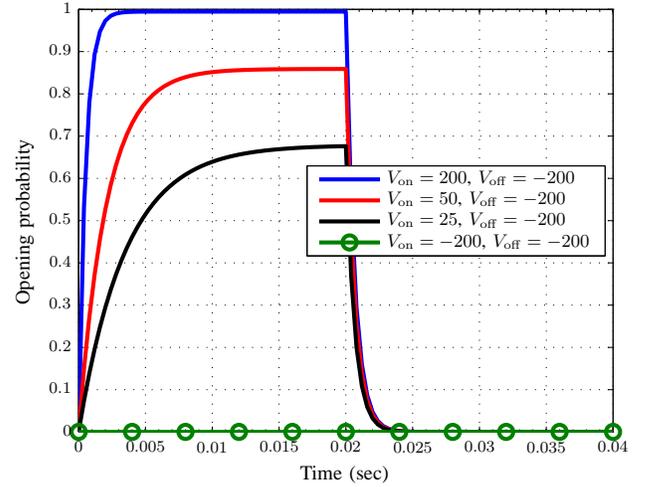}} 
\vspace{-1cm}
\caption{Opening probability for voltage $\mathcal{V}^1(t)$ for different $V_{\rm on}$ values and $V_{\rm off}=-200$ mv.}\vspace{-0.5cm}
\label{fig0_Opening_prob}
\end{figure}
To analyze the performance of the considered modulation scheme, the release rate of the molecules from the cell given a desired opening probability signal (modulated signal) has to be studied. 
Given the desired opening probability for bit 1, $\mathcal{P}_o^1(t)=u(t)-u(t-T_1)$ or equivalently a corresponding voltage signal $\mathcal{V}^1(t)$, all ion channels are open during $[0,T_1]$. Therefore, the molecules produced inside the cell gradually move out  because of the force caused by the concentration gradient between the inside and outside the cell. Thereby, the molecules inside the cell are not released instantaneously, i.e., the release rate of the molecules (modulated signal) is not a Dirac delta function. In the next section, we derive the average release rate of the molecules (i.e., the average modulated signal) for the considered IBM.      

         \vspace{-0.5cm}

\section{Average modulated signal of IBM }
In this section. the average modulated signal for IBM based on-off keying modulation is  derived from the corresponding diffusion problem. Thereby, an analytical expression for the Laplace transform of the average modulated signal is obtained. Furthermore, we derive a closed-form upper bound on the total number of released molecules. 
\vspace{-0.3cm}
\subsection{Diffusion Problem Formulation for IBM} 
Given $\mathcal{V}^1(t)=V_{\rm on} (u(t)-u(t-T_1))+V_{\rm off}(u(t-T_1)-u(t-T))$, ideally the opening probability of the ion channels is given by $P_o^{\infty}(V_{\rm on})=1$ in the interval $[0,T_1]$. We define the average permeability of the cell surface, $z$, as the ratio of the average area of the open ion channels to the membrane area which yields $z=\frac{P_o^{\infty}(V_{\rm on})N \pi r_c^2}{4\pi r_m^2}$. For sufficiently large numbers of ion channels, the membrane can be considered uniformly permeable, i.e., if an $A$ ion hits the membrane, it leaves the cell with probability $z$ and is reflected with probability $1-z$. 

The $A$ ion concentrations inside and outside the cell are denoted by $C_1(r,t), r \leq r_m$, and $C_2(r,t), r>r_m$, respectively. While the channels are open and the molecules move out, the molecule source is producing new molecules at constant rate $\mathcal{S} \delta(r-r_s)$ where $\delta(\cdot)$ denotes the Dirac delta function. 
Given the source $\mathcal{S} \delta(r-r_s)$ inside the cell, we can formulate the diffusion equation for the concentrations of molecules inside and outside the cell in the interval $[0,T_1]$ as follows \cite{Crank1955}
\vspace{-0.2 cm}
\begin{align}
&\frac{\partial C_1(r,t)}{\partial t}=D_1 \frac{1}{r^2} \frac{\partial}{\partial r} \left(\frac{r^2 \partial C_1(r,t)}{\partial r}\right)\label{diffeq1}\\
&\frac{\partial C_2(r,t)}{\partial t}=D_2 \frac{1}{r^2} \frac{\partial}{\partial r} \left(\frac{r^2 \partial C_2(r,t)}{\partial r}\right)\label{diffeq2}.  
\end{align} 
Because of the spatial symmetry of the problem only the derivatives with respect to $r$ appear in the diffusion equations \eqref{diffeq1} and \eqref{diffeq2}.     
The initial concentrations inside and outside the cell in the considered time slot are $\mathcal{T}$ and zero, respectively, i.e, $C_1(r,0)=\mathcal{T},\; r \leq r_m$, and $C_2(r,0)=0,\;  r > r_m$, respectively, regardless of whether the previously transmitted bit is 0 or 1. In fact, this is achieved by closing the ion channels during the last $T_2$ seconds of the preceding time slot. 

In addition to the initial concentrations, the boundary conditions over the membrane are needed to solve the above diffusion problem.
A general method for obtaining boundary conditions for diffusion problems has been reported in \cite{Naqvi1982, Chatur1983}. Furthermore, exploiting the methods provided in these papers, the authors in \cite{Koszto2001} derive the boundary conditions for a membrane for a one dimensional diffusion problem. In the following proposition, we extend the result in \cite{Koszto2001} to 3-dimensional diffusion over a spherical membrane  at $r=r_m$. 
\begin{proposition} The boundary conditions for the membrane of the considered IBM are given by
\begin{align}\label{BC1}
&D_1\frac{\partial C_1(r,t)}{\partial r}=D_2\frac{\partial C_2(r,t)}{\partial r},\quad r=r_m,\\
&D_1\frac{\partial C_1(r,t)}{\partial r}=\frac{z}{1-z}\sqrt {\frac{k_B \mathfrak{T}}{2\pi m}}(C_2(r,t)-C_1(r,t)), r=r_m,\nonumber
\end{align} 
where $k_B$ is the Boltzman constant, $\mathfrak{T}$ is the temperature, and $m$ is the mass of the $A$ molecule.
\end{proposition}
\begin{proof}
Please refer to the Appendix.
\end{proof}
Moreover, a trivial boundary condition for the considered problem is  $C_2(\infty,t)=0$. Please note that we have not considered any boundary condition for $r=r_s$. In fact as $r_s\ll r_m$, the effect of the molecule generator volume can be neglected for the diffusion problem. 

Solving the considered diffusion problem becomes easier, if we can resolve the nonzero initial condition inside the cell. An initial concentration at radius $r_i$, $C_1(r_i,0)$, is equivalent to having an instantaneous molecule production source at time 0 at radius $r_i$, i.e, $C_1(r_i,0)\delta(r-r_i)\delta(t)$. Therefore, the constant initial concentration, $\mathcal{T}$, inside the cell, $r \leq r_m$, can be modeled as an instantaneous source with $\mathcal{T}\delta(t)$ for $r \leq r_m$ which can be combined with the molecule generator source $\mathcal{S} \delta(r-r_s)$. This yields the following composite source inside the cell:
\begin{align}\label{integ_source}
S(r,t)=\mathcal{S}\delta(r-r_s)+\mathcal{T}\delta(t) \quad r \leq r_m, \; 0\leq t \leq T_1,
\end{align}
and zero initial conditions. Given $C_2(r,t)$, the average diffusion flux at surface $r=r_m$ equals $-D_2 \frac{\partial C_2(r,t)}{\partial r}=-D_1 \frac{\partial C_1(r,t)}{\partial r},\;r=r_m$ \cite{Crank1955}. Therefore, the average release rate of the molecules from the cell across the entire membrane area, i.e., the average modulated signal is given by
\begin{align}\label{exitrate1}
w(t)=- 4\pi r_m^2 D_2 \frac{\partial C_2(r,t)}{\partial r} \quad r=r_m.
\end{align}
Summarizing the discussion above, the average modulated signal is the solution of the following problem.
\begin{problem}\hspace{-0.3cm}\footnote{We note that the author in \cite{Mild1972}, considers Problem \ref{problem1} for the special case of $\mathcal{S}=0$.}\label{problem1} Given the modulating voltage signal $\mathcal{V}^1(t)=V_{\rm on} (u(t)-u(t-T_1))+V_{\rm off}(u(t-T_1)-u(t-T))$, the average modulated signal, $w(t)$, in the interval $[0,T_1]$ is the solution of the differential equations
\begin{align} \label{problem}
&w(t)=- 4\pi r_m^2 D_2 \frac{\partial C_2(r,t)}{\partial r} \quad r=r_m,\\
 &\frac{\partial C_1(r,t)}{\partial t}=D_1 \frac{1}{r^2} \frac{\partial}{\partial r} \left(\frac{r^2 \partial C_1(r,t)}{\partial r}\right)+S(r,t) \label{problem1eq2}\quad r\leq r_m,\\
 &\frac{\partial C_2(r,t)}{\partial t}=D_2 \frac{1}{r^2} \frac{\partial}{\partial r} \left(\frac{r^2 \partial C_2(r,t)}{\partial r}\right) \quad r>r_m, \label{problem1eq3}
\end{align}  \normalsize
for boundary conditions \eqref{BC1} and $C_2(\infty,t)=0$, and zero initial conditions inside and outside the cell. Furthermore, $S(r,t)$ in \eqref{problem1eq2} is given by \eqref{integ_source}.
\end{problem}
Considering the transmitter as a system with input signal $S(r,t)$ and output signal $w(t)$, it is clear that the system is linear and time invariant. As a result, the average modulated signal, $w(t)$, can be expressed as  
\begin{align}\label{exitrate}
w(t)&=\int_{0}^{r_m} \int_{0}^{t} { S(r',t')w^\star(t-t'|r')dt'dr'}\nonumber \\
&= \int_{0}^{r_m} { S(r',t)\star w^\star(t|r')dr'},
\end{align}
where $\star$ denotes the convolution operation and $w^\star (t|r')$ denotes the system response to the impulse $\delta(r-r')\delta(t)$ and is referred to as modulator impulse response. 
\vspace{-0.3cm}
\subsection{Modulator Impulse Response}
 
The impulse response $w^\star (t|r')$ is the solution of Problem \ref{problem1} for $S(r,t)=\delta(r-r')\delta(t)$. Let us define the dimensionless radial coordinate, time, and diffusion coefficient as $\rho=r/r_m,\; 
\tau=D_1t/r_m^2,$ and $A=\frac{D_2}{D_1}$, respectively, and $h=\frac{r_m}{D_1}\frac{z}{1-z}\sqrt {\frac{k_B \mathfrak{T}}{2\pi m}}$, and $\rho'=r'/r_m$.
Employing a change of variables, namely
\begin{align}
U_1(\rho,\tau)&=\rho C_1\left(r_m\rho,\frac{r_m^2\tau}{D_1}\right),\;
U_2(\rho,\tau)=\rho C_2\left(r_m\rho,\frac{r_m^2\tau}{D_1}\right),\nonumber\\
\varphi^\star(\tau|\rho') &=\frac{r_m^2}{D_1}w^\star\left(\frac{r_m^2\tau}{D_1}|\rho'r_m\right),\nonumber
\end{align} 
Problem \ref{problem1} is transformed to the following problem which includes only diffusion equations with constant coefficients.
\begin{problem} \label{problem2}
The modulator impulse response in dimentionless variables $\tau$ and $\rho$, $\varphi^\star (\tau|\rho')$, is obtained from the following system of differential equations 
\begin{align}
&\varphi^\star(\tau|\rho')=-{4\pi r_m^3 A} \left(\frac{\partial U_2(\rho,\tau)}{\partial \rho}-U_2(\rho,\tau)\right), \; \rho=1 \label{WintermsU}, \\
&\frac{\partial U_1(\rho,\tau)}{\partial \tau}= \frac{\partial ^2 U_1(\rho,\tau)}{\partial \rho^2}+\frac{\rho'}{r_m}\delta(\rho-\rho')\delta(\tau),\; \rho\leq 1 \label{diffeq3}\\
&\frac{\partial U_2(\rho,\tau)}{\partial \tau}= A \frac{\partial ^2 U_2(\rho,\tau)}{\partial \rho^2}\quad \rho>1 \label{diffeq32}
\end{align}\normalsize
where the following boundary and initial conditions hold:
 \begin{align}
&\frac{\partial U_1(\rho,\tau)}{\partial \rho}-U_1(\rho,\tau)=A\left(\frac{\partial U_2(\rho,\tau)}{\partial \rho}-U_2(\rho,\tau)\right),\; \rho=1  \label{BC21}\\
&\frac{\partial U_1(\rho,\tau)}{\partial \rho}+(h -1)U_1(\rho,\tau)=h U_2(\rho,\tau),\qquad \rho=1 \label{BC22}\\
&U_2(\rho,\tau)=0, \qquad \rho \to \infty \label{BC23}\\
&U_1(\rho,\tau)=0, \qquad \rho \to 0 \label{BC24}\\
&U_1(\rho,0)=0, \qquad \rho\leq 1 \label{IC21}\\
&U_2(\rho,0)=0, \qquad \rho>1. \label{IC22}
\end{align}\normalsize
Here, (\ref{BC21}) and (\ref{BC22}) are the transformed boundary conditions in  (\ref{BC1}), (\ref{BC23}) was obtained by transforming $C_2(\infty,t)=0 $, (\ref{BC24}) holds since $U_1(\rho,\tau)=\rho C_1$, and (\ref{IC21}) and (\ref{IC22}) are the transforms of the zero initial conditions. 
\end{problem}
The Laplace transformation with respect to $\tau$ is employed to solve Problem \ref{problem2} \cite{Carslaw1959,Crank1955}. Taking the Laplace transform of \eqref{WintermsU}, \eqref{diffeq3}, and  \eqref{diffeq32} results in a system of ordinary differential equations which can be easily solved. This leads to  
\begin{align}\label{exitrateeq}
\overline{\varphi}^\star(s|\rho')={4\pi r_m^3 A}\left(1+\sqrt{\frac{s}{A}}\right) \overline{U}_2(\rho,s), \quad \rho=1,
\end{align}
where $\overline{(\cdot)}$ denotes the Laplace transform of $(\cdot)$ and $\overline{U}_2(1,s)$ is given by \eqref{upper_con} at the top of the next page.
\begin{table*}
\begin{align}\label{upper_con}
\overline{U}_2(1,s)=\frac{-h\rho'}{r_m}\frac{\sinh(\sqrt{s}\rho')}{(h-(h-1)A (1+\sqrt{s/A}))\sinh{(\sqrt{s})}- (h+A(1+\sqrt{s/A}))\sqrt{s} \cosh{(\sqrt{s})} }.
\end{align}
\vspace{-0.7cm}
\hrulefill
\end{table*}
In general, obtaining the inverse Laplace transform of ${\overline{\varphi}}\,^\star(s|\rho')$ in closed form does not seem possible by standard methods and numerical inversion is needed, see e.g. \cite{Mild1972}. 
However, adopting a simplifying assumption, a closed-form expressions for the average modulated signal can be obtained which leads to an upper bound on the total number of released molecules.

%
%
\vspace{-0.3cm}
\subsection{Upper Bound on Total Number of Released Molecules}\label{upperbound}
The average modulated signal is the solution of Problem \ref{problem1} and is given by \eqref{exitrate}. The average number of molecules that leave the cell until time $t$ is the integral over the average modulated signal with respect to $t$, i.e.,
\begin{align}\label{integral_w}
M(t)=\int_{0}^{t}{w(\beta) d\beta}.
\end{align}

To obtain an upper bound on $M(t)$, the concentration outside the cell is assumed to be zero at all times, i.e., $C_2(r,t)=0$, $r>r_m$, $t\in [0,T]$. The physical interpretation of zero concentration outside the cell is that no molecule outside the cell can come back inside the cell, i.e., the molecules that leave the cell disappear and never return. Therefore, the net number of molecules released from the cell under this assumption is higher than the actual number of released molecules. As a result, this assumption yields an upper bound on the average number of released molecules, $M(t)$. 

Given the assumption $C_2(r,t)=0$ or equivalently $U_2(\rho,\tau)=0$, Problem \ref{problem2} reduces to a well known heat conduction problem with a closed-form solution for $U_1(\rho,\tau)$ given by \cite[Page 237]{Carslaw1959}
\begin{align}
U_1(\rho,\tau)=\frac{2\rho'}{r_m} \sum_{n=1}^{\infty}{e^{-\gamma_n^2\tau}\sin(\gamma_n\rho)\sin(\gamma_n\rho')\frac{\gamma_n^2+(h-1)^2} {\gamma_n^2+h(h-1)}}\nonumber
\end{align} 
where $\mp\gamma_n, n=1,2,\cdots,$ are the roots of
\begin{align}\label{cotroots}
\gamma \cot(\gamma)+h-1=0.
\end{align} 

The modulator impulse response given the upper bound assumption is denoted by $w_u^\star(t|r')$ and given by $-4\pi r_m^2 D_1 \frac{\partial C_1(r,t)} {\partial r}$, $r=r_m$. In dimensionless variables $\tau$ and $\rho'$, $w_u^\star(t|r')$ is denoted by $\varphi_u^\star(\tau|\rho')$ and is equal to $-{4\pi r_m^3} (\frac{\partial U_1(\rho,\tau)}{\partial \rho}-U_1(\rho,\tau))$, i.e.,
\vspace{-0.5cm}
\begin{align}\label{upper_impulse}
\varphi_u^\star(\tau|\rho')= \sum_{n=1}^{\infty}{G_n \rho'e^{-\gamma_n^2\tau}\sin(\gamma_n\rho')},
\end{align} 
where 
\begin{align}
G_n&=-8\pi r_m^2(\gamma_n\cos(\gamma_n)-\sin(\gamma_n))\frac{\gamma_n^2+(h-1)^2} {\gamma_n^2+h(h-1)}\\
&=8\pi r_m^2h\sin(\gamma_n)\frac{\gamma_n^2+(h-1)^2} {\gamma_n^2+h(h-1)}.
\end{align} 
%
%
The corresponding average modulated signal in the interval $[0,T_1]$ is
\begin{align}
w_u(t)=\int_{0}^{r_m} \int_{0}^{t} { S(r',t')w_u^\star(t-t'|r')dr'dt'}.
\end{align}
As a result, $w_u(t)$ and the upper bound on $M(t)$, $M_u(t)=\int_{0}^{t}{w_u(\beta)d\beta}$, are given by the closed-form expressions in \eqref{upper_modulated} and \eqref{upper_integral} at the top of the next page.
\begin{table*}[t]
\begin{align} 
w_u(t)&=
\sum_{n=1}^{\infty} {G_n\left[\frac{\mathcal{S} r_s}{\gamma_n^2r_m}  \left(1-e^{-\frac{\gamma_n^2D_1t}{r_m^2}} \right)\sin\left(\frac{\gamma_n r_s}{r_m}\right) +{ \frac{\mathcal{T}hD_1}{\gamma_n^2r_m}} e^{-\frac{\gamma_n^2 D_1 t}{r_m^2}}\sin(\gamma_n)  \right]}\label{upper_modulated}\\
M_u(t)&=
\sum_{n=1}^{\infty} {G_n\left[\frac{\mathcal{S} r_s}{\gamma_n^2r_m} \frac{r_m^2}{D_1} \left(t-\frac{1}{\gamma_n^2}(1-e^{-\frac{\gamma_n^2D_1t}{r_m^2}}) \right)\sin\left(\frac{\gamma_n r_s}{r_m}\right) +{ \frac{\mathcal{T}hD_1}{\gamma_n^2r_m}} \frac{r_m^2}{D_1\gamma_n^2}(1-e^{-\frac{\gamma_n^2 D_1 t}{r_m^2}})\sin(\gamma_n)  \right]}\label{upper_integral}
\end{align}
\vspace{-0.7cm}
\hrulefill
\end{table*}

\vspace{-0.3cm}
  
\section{Numerical and simulation results}

For the numerical and simulation results, we consider a transmitter of radius $r_m=5$ $\mu$m, $r_s=0.5$ $\mu$m, and $N$ Potassium ($K^+$) ion channels with $r_c=1$ nm (modeled as described in Example \ref{example1}) distributed uniformly over the surface of the transmitter. The internal and external environments of the transmitter are assumed be water at temperature 27$^\circ C$ with diffusion coefficient $D_1=D_2=1.14\times 10^{-9}$ m$^3/$s for Potassium. For all results, we assume $V_{\rm on}=200$ mv and $V_{\rm off}=-200$ mv, and correspondingly, $P_o(V_{\rm on})=1$ and $P_o(V_{\rm off})=0$. The molecule generator source is assumed to have a generation rate of $\mathcal{S}=3\times10^{14}$ mo$/$(m$^2$s) and a stop concentration threshold of $\mathcal{T}=10^{18}$ mo$/$(m$^3$).

\begin{figure}
  \centering
\resizebox{1\linewidth}{!}{\psfragfig{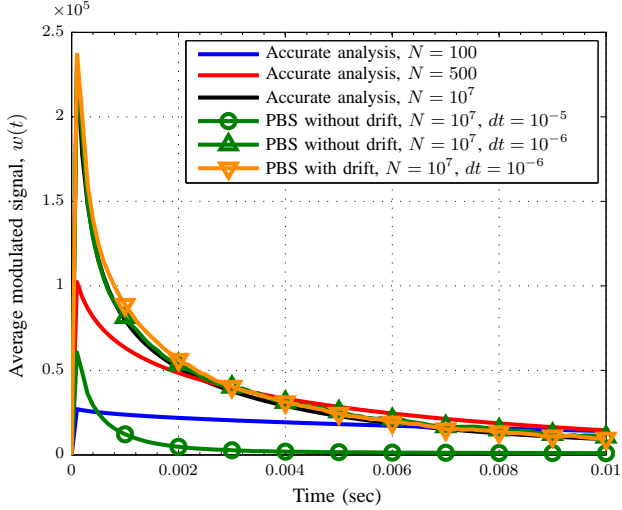}} 
\vspace{-1cm}
\caption{Average modulated signal obtained from the proposed accurate analysis and a PBS.  \vspace{-2cm}}
\label{ACCPBS}
\end{figure}
\begin{figure}
\vspace{-0.4cm}
  \centering
\resizebox{1\linewidth}{!}{\psfragfig{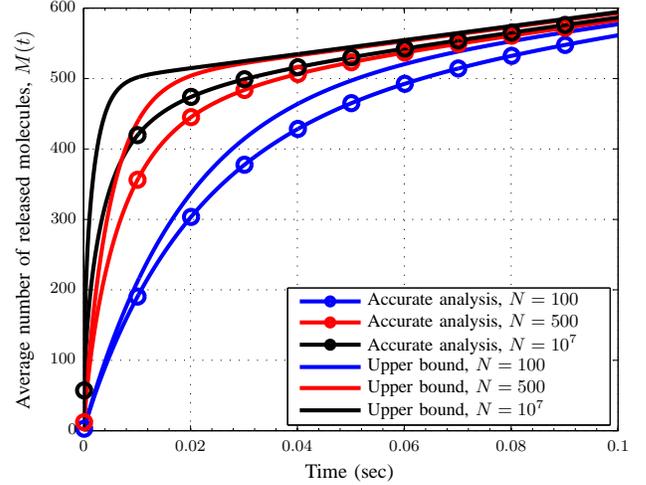}} 
\vspace{-1cm}
\caption{Comparison of accurate analysis of average number of released molecules with upper bound for different numbers of ion channels.}
\label{ACCUP}
\end{figure}
   
To confirm the analysis of the modulated signal, we use a PBS \cite{R1}.
In the PBS, the molecule locations are known and the molecules move independently in the 3-dimensional space. In each dimension, the displacement of a molecule in $dt$ seconds is modeled as a Gaussian random variable (rv) with zero mean and variance $2Ddt$. If a molecule hits the membrane, the outcome of a Bernoulli experiment with success probability $z$ is conducted and depending on the outcome of the experiment the molecule passes the membrane or is reflected. 

\vspace{-0.1cm}
In Fig. \ref{ACCPBS}, the average modulated signal obtained from the analysis, $w(t)$, is shown for different numbers of ion channels, i.e., $N=100, 500$, and $10^7$. In particular, to evaluate the modulator impulse response in \eqref{exitrate}, the inverse Laplace transform of ${\overline{\varphi}}\,^\star(s|r)$ in \eqref{exitrateeq} is computed numerically using the Talbot's method. Fig. \ref{ACCPBS} reveals that the molecules exit gradually from the cell. Thereby, the average modulated signal approaches to Dirac delta impulse as the number of ion channels increases. The average modulated signal obtained from the PBS is also depicted for $N=10^7$ and two different time steps for random walk of the molecules, i.e., $dt=10^{-5},$ and $dt=10^{-6}$. It is observed that by choosing a sufficiently small $dt$ (here $dt=10^{-6}$) the result obtained with the PBS approaches the analytical result. However, for $dt=10^{-5}$, the PBS result deviates from the analytical result. This can be explained as follows. The molecule movements have a continuous trajectory. Hence, in order to obtain a sufficient accurate trajectory in the simulation, a small time step, $dt$, is needed. Specifically, for membranes with smaller permeability $z$ (i.e., a smaller number of ion channels), smaller $dt$ values are needed to accurately track the molecule release rate.
Furthermore, the effect of electrical drift of the ions caused by their charge is investigated by the PBS for $N=10^7$. Thereby, Fig. \ref{ACCPBS} reveals that the effect of electrical drift is negligible for the considered typical cell size.


In Fig. \ref{ACCUP}, the average number of released molecules obtained from the accurate analysis, $M(t)$, is compared to the upper bound, $M_u(t)$, for different numbers of ion channels. It is observed that for smaller numbers of ion channels the upper bound is tighter. Since smaller numbers of ion channels result in smaller release rates of molecules, the concentration of the molecules outside the cell is smaller and the assumption made to arrive at the upper bound is more justified.

Note that for Figs. \ref{ACCPBS} and \ref{ACCUP}, we have assumed the ideal desired opening probability, $\mathcal{P}^1(t)$. However, the actual ion channel opening probability has a non-zero time constant as shown in Fig. \ref{fig0_Opening_prob}. For example, for the values of $V_{\rm on}=200$ mv and $V_{\rm off}=-200$ mv, the time constant of the opening probability is approximately $0.5$ ms and it takes $4\times 0.5$ ms to reach 95 percent of the final value. Hence, compared to the rise time of the accurate analysis curves for the molecule release in Fig. \ref{ACCUP}, the time constant of the opening probability is negligible, especially for smaller numbers of ion channels, and therefore the assumption is justified. 
\vspace{-0.5cm}
\section{Conclusions}
In this paper, an IBM for DMC was proposed. The IBM controls the rate at which molecules are released at the transmitter by modulating a gating parameter. Furthermore, a simple IBM based on-off keying modulation was analyzed and the Laplace transform of the average modulated signal was derived. Adopting the simplifying assumption of zero molecule concentration outside the cell, a closed-form expression for the modulated signal was presented and shown to correspond to an upper bound on the total number of released molecules. This bound is tight for a small number of ion channels.
Our results reveal that a real modulator releases molecules gradually and not instantaneously as is often assumed in the literature. 

In this paper, we only considered the average modulated signal, while the actual modulated signal displays a stochastic behavior because of the Brownian motion of the molecules and deviates from the average. Analyzing the impact of the stochastic behavior of the modulated signal on the performance of IBM based transmission is an interesting topic for future work.      

\begin{appendix}
In this appendix, we present a proof for Proposition 1. Before starting the proof, we define the normalized distribution function of a free Brownian particle. Moreover, we provide an important property of this function which is utilized to prove the Proposition.   
\begin{definition} The normalized distribution function of a free Brownian particle in one-dimension is denoted by $f(V_x,X,t)$ and is defined as the probability that the particle at time $t$ is found at location $X$ with velocity $V_x$. $f(V_x,X,t)$ can be obtained from the Fokker-Planck equation \cite{Naqvi1982, Chatur1983}.
\end{definition}
\begin{property}
Function $f(V_x,X,t)$ can be written as a summation of an even function $f_0(V_x,X,t)$, and an odd function $f_1(V_x,X,t)$ with respect to $V_x$ as follows \cite{Naqvi1982},
\begin{align}\label{disfunc1}
f(V_x,X,t)=f_0(V_x,X,t)+f_1(V_x,X,t)
\end{align}
where
\begin{align}
&f_0(V_x=v_x,X=x,t)=C_X(x,t)\mathcal{N}_{V_x}(v_x,k_B\mathfrak{T}/m)\\
&f_1(V_x,X,t)=-f_1(-V_x,X,t),
\end{align}
where $k_B$ is the Boltzman constant, $\mathfrak{T}$ is the temperature, $m$ is the particle mass, $\mathcal{N}_{V_x}(v_x,k_B\mathfrak{T}/m)$ denotes the normal probability distribution function (pdf) of random variable $V_x$ with variance $k_B\mathfrak{T}/m$, i.e., $$\mathcal{N}_{V_x}(v_x,k_B\mathfrak{T}/m)=\frac{1}{\sqrt{2\pi k_B\mathfrak{T}/m}}\exp\left(\frac{-v_x^2}{2{k_B\mathfrak{T}/m}}\right),$$ and $C_X(x,t)$ denotes the pdf of random variable $X$ which is interpreted as the normalized concentration at location $x$ and time $t$. Note that $f_0(V_x,X,t)$ is the equilibrium distribution function of the particle and is the dominant contributing term in the distribution function, i.e., $f_0(v_x,x,t)\gg |f_1(v_x,x,t)|,$ $\forall$ $v_x,x,$ and $t$ \cite{Koszto2001}.
\end{property}
The authors in \cite{Koszto2001}, consider one-dimensional diffusion, where there is a permeable membrane at $X=0$ with permeability $z$. Exploiting Property 1, they show that the boundary conditions over the membrane are as follows\footnote{In \cite{Koszto2001}, the diffusion coefficients of the left and the right hand side of the membrane are assumed to be equal, $D_1=D_2$, but the presented proof holds for the general case of $D_1\neq D_2$.}:
\begin{align}
&D_1\frac{\partial C_1(x,t)}{\partial x}=D_2\frac{\partial C_2(x,t)}{\partial x},\quad x=0,\\
&D_1\frac{\partial C_1(x,t)}{\partial x}=\frac{z}{1-z}\sqrt {\frac{k_B \mathfrak{T}}{2\pi m}}(C_2(x,t)-C_1(x,t)), x=0,\nonumber
\end{align}
where $C_1(x,t)$ and $C_2(x,t)$ are the concentrations for $x\leq 0$ and $x > 0$, respectively. In our problem, we consider the Brownian motion in a 3-dimensional environment. The membrane with permeability $z$ is assumed to be a spherical surface at $R=r_m$ where $R$ is the radial coordinate of the spherical coordinate system. Assume a point $\textbf{P}=(x,y,z)$ in Cartesian coordinates or equivalently $\textbf{P}=(r,\theta,\phi)$ in spherical coordinates, where $\theta$ and $\phi$ are the polar and azimuth angles, respectively. The unit vector in radial direction at this point is given by 
\begin{align}
\hat{a}_r=\cos(\phi)\sin(\theta)\hat{a}_x+\sin(\phi)\sin(\theta)\hat{a}_y+\cos(\theta)\hat{a}_z,
\end{align} 
where $\hat{a}_x,\hat{a}_y,$ and $\hat{a}_z$ are unit vectors in directions $X,Y,$ and $Z$, respectively. Defining the velocity vector of the particle as $\textbf{V}=v_x\hat{a}_x+v_y\hat{a}_y+v_z\hat{a}_z$, the component of the velocity vector of the particle in direction $\hat{a}_r$ is defined as $v_r=\langle V, \hat{a}_r\rangle$, where $\langle \cdot,\cdot \rangle$ denotes the inner product operation. As a result, the radial component of the velocity of a particle at a given point $\textbf{P}$ is 
\begin{align}\label{vr1}
v_r=v_x\cos(\phi)\sin(\theta)+v_y\sin(\phi)\sin(\theta)+v_z\cos(\theta).
\end{align} 
To prove the proposition, we need to derive the joint distribution function of the location and the radial velocity of a free Brownian particle, i.e., $f(v_r,\textbf{P},t)$, since the boundary is at $R=r_m$. We show that $f(v_r,\textbf{P},t)$ can be approximated as the summation of an even function, $C(\textbf{P},t)\mathcal{N}_{V_r}(v_r,k_B\mathfrak{T}/m)$, and an odd function,  $g(v_r,\textbf{P},t)$, with respect to $v_r$ as follows: 
\begin{align}\label{disfunc21}
f(v_r,\textbf{P},t)=C(\textbf{P},t)\mathcal{N}_{V_r}(v_r,k_B\mathfrak{T}/m)+g(v_r,\textbf{P},t).
\end{align}
where $C(\textbf{P},t)$ denotes the pdf of location of the particle. Distribution function \eqref{disfunc21} can be interpreted as the distribution function in a one-dimensional environment as given by \eqref{disfunc1} but with respect to $r$. Thereby, the rest of the proof for deriving the boundary conditions for a 3-dimensional environment is similar to the proof in \cite{Koszto2001}. Therefore, in the remainder of the proof, we derive the distribution function $f(v_r,r,t)$ as given in \eqref{disfunc21}. 

The Brownian motion is independent in the coordinates $X, Y,$ and $Z$. Thereby, we can write the normalized distribution function for the location and velocity of the particle in $Y$ and $Z$ coordinates, respectively, also as in \eqref{disfunc1}. 
Considering \eqref{disfunc1}, the marginal pdfs of the velocity and the location in the $X$ coordinate can be obtained as
\begin{align}
C_X(X,t)=\int{f(V_x,X,t)dV_x},\\
\xi_{V_x}(V_x,t)=\int{f(V_x,X,t)dX},
\end{align}
where the indefinite integrals are over the entire location space and the entire velocity space, respectively. Given the location of the particle at $x$ and at time $t$, the conditional pdf of the velocity is \cite{Feller1968}
\begin{align}
\xi_{V_x}(V_x|X=x,t)=\frac{f(V_x,X=x,t)}{C_X(X=x,t)}.
\end{align}
Therefore, we have
\begin{align}\label{condpdf}
\xi_{V_x}(V_x=v_x|X=x,t)=\mathcal{N}_{V_x}(v_x,k_B\mathfrak{T}/m)+\frac{f_1(v_x,x,t)}{C_X(x,t)}.
\end{align}
Similarly, pdfs $\xi_{V_y}(V_y|Y,t)$ and $\xi_{V_z}(V_z|Z,t)$ can be obtained for the velocity components in the $Y$ and $Z$ coordinates, respectively.
The velocity component in radial direction at a given point $\textbf{P}$ is given by \eqref{vr1}.
Defining random variable ${V_x^r=V_x \cos(\phi)\sin(\theta)}$ and given the pdf $\xi_{V_x}(V_x|x,t)$, the pdf of random variable $V_x^r$ can be written as \cite{Feller1968}
\begin{align}
\xi_{V_x^r}(V_x^r|\textbf{P},t)&=\frac{1}{\cos(\phi)\sin(\theta)} \xi_{V_x}\left(\frac{V_x^r}{\cos(\phi)\sin(\theta)}|\textbf{P},t\right)\nonumber\\
&\overset{(a)}{=} \frac{1}{\cos(\phi)\sin(\theta)} \xi_{V_x}\left(\frac{V_x^r}{\cos(\phi)\sin(\theta)}|x,t\right)
\end{align}
where equality (a) holds since the Brownian motion in the $X$ coordinate is independent from the $Y$ and $Z$ coordinates, and as a result, $\xi_{V_x}(V_x|\textbf{P}=(x,y,z),t)=\xi_{V_x}(V_x|x,t)$.    
Considering $\xi_{V_x}(V_x|x,t)$ as given in \eqref{condpdf}, we have
\begin{align}\label{vx}
&\xi_{V_x^r}(V_x^r=v_x^r|\textbf{P},t)=
\mathcal{N}_{V_x^r}(v_x^r,(\cos(\phi)\sin(\theta))^2 k_B \mathfrak{T}/m)+\\
&\frac{1}{C_X(x,t)\cos(\phi)\sin(\theta)} f_1\left(\frac{v_x^r}{\cos(\phi)\sin(\theta)},x,t\right),\nonumber
\end{align}
where the first term is even and the second term is odd with respect to $v_x^r$.
The distributions $\xi_{V_y^r}(V_y^r|\textbf{P},t)$ and $\xi_{V_z^r}(V_z^r|\textbf{P},t)$ for $V_y^r$ and $V_z^r$ can be obtained in a similar manner as follows:
\begin{align}
&\xi_{V_y^r}(V_y^r=v_y^r|\textbf{P},t)=
\mathcal{N}_{v_y^r}(v_y^r,(\sin(\phi)\sin(\theta))^2 k_B \mathfrak{T}/m)\label{vy}\\
&+\frac{1}{C_Y(y,t)\sin(\phi)\sin(\theta)} f_1\left(\frac{v_y^r}{\sin(\phi)\sin(\theta)},y,t\right),\nonumber\label{vz}\\
&\xi_{V_z^r}(V_z^r=v_z^r|\textbf{P},t)=
\mathcal{N}_{V_z^r}(v_z^r,(\sin(\phi)\sin(\theta))^2 k_B \mathfrak{T}/m)\\
&+\frac{1}{C_Z(z,t)\cos(\theta)} f_1\left(\frac{v_z^r}{\cos(\theta)},z,t\right).\nonumber
\end{align}
For a given point $\textbf{P}$, the rvs $V_x^r$, $V_y^r$, and $V_z^r$ are independent from each other. Therefore, the distribution function of $V_r=V_x^r+V_y^r+V_z^r$, given point $\textbf{P}$ at time $t$, $\xi_{V_r}(V_r|\textbf{P},t)$, is given by the convolution of the distribution functions of $V_x^r,V_y^r$, and $V_z^r$ \cite{Feller1968}, i.e.,
\begin{align}\label{vr}
\xi_{V_r}(V_r=v_r|\textbf{P},t) = \xi_{V_x^r}(v_r|\textbf{P},t)  \star \xi_{V_y^r}(v_r|\textbf{P},t)\star \xi_{V_z^r}(v_r|\textbf{P},t).
\end{align}
Substituting \eqref{vx}, \eqref{vy}, and \eqref{vz} into \eqref{vr} results in the summation of 8 convolution terms. The first term is the convolution of 3 normal pdfs as follows: 
\begin{align}\label{convterm1}
\mathcal{N}_{V_x^r}(v_r,(\cos(\phi)\sin(\theta))^2 k_B \mathfrak{T}/m) \star\\
 \mathcal{N}_{V_y^r}(v_r,(\sin(\phi)\sin(\theta))^2 k_B \mathfrak{T}/m) \star\nonumber\\ \mathcal{N}_{V_z^r}(v_r,(\sin(\phi)\sin(\theta))^2 k_B \mathfrak{T}/m)\nonumber
\end{align}
which results in a normal pdf with a variance equal to the summation of the individual variances, i.e., 
\begin{align}
&(\cos(\phi)\sin(\theta))^2 k_B \mathfrak{T}/m+(\sin(\phi)\sin(\theta))^2 k_B \mathfrak{T}/m+\\
&\cos(\theta)^2 k_B \mathfrak{T}/m=k_B \mathfrak{T}/m.\nonumber
\end{align}
In other words, the summation of three independent normal random variables is a new normal random variable whose variance is the summation of the variances of the constituting rvs.
Therefore, the first term is $\mathcal{N}_{v_r}(k_B\mathfrak{T}/m)$ which is an even function.

There are 3 (of 8) other even functions which are convolutions of one normal pdf and two functions containing $f_1$\footnote{It is easy to show that the convolution of an odd function with one even function is odd and the convolution of two odd or two even function is even.} For example, one of these terms is 
\begin{align}\label{convterm2}
\mathcal{N}_{V_x^r}(v_r,(\cos(\phi)\sin(\theta))^2 k_B \mathfrak{T}/m)\star\\
 \frac{1}{C_Y(y,t)\sin(\phi)\sin(\theta)} f_1\left(\frac{v_r}{\sin(\phi)\sin(\theta)},y,t\right)\star \nonumber\\
  \frac{1}{C_Z(z,t)\cos(\theta)} f_1\left(\frac{v_r}{\cos(\theta)},z,t\right).\nonumber
\end{align}
Since $f_0(v_y,y,t)\gg |f_1(v_y,y,t)|$ and $f_0(v_z,z,t)\gg |f_1(v_z,z,t)|$, it is easy to see that the second and third terms in \eqref{convterm2} are much smaller than the second and third terms in \eqref{convterm1}, respectively. As a result, the term in \eqref{convterm2} is negligible in comparison to that in \eqref{convterm1}. Similarly, the other two even functions can be shown to be negligible in comparison to \eqref{convterm1}.  

Moreover, there are 3 (of 8) functions that are the convolution of two normal pdfs (even function) and one odd function and there is one (of 8) function that is the convolution of 3 odd functions. The summation of these 4 functions is odd. Therefore, we can write $\xi_{V_r}(V_r=v_r|\textbf{P},t)$ as summation of an even function $\mathcal{N}_{V_r}(v_r,k_B\mathfrak{T}/m)$ and an odd function. Thereby, the joint distribution function of the radial velocity, $V_r$, and the location $\textbf{P}$ of a free Brownian particle in a 3-dimensional environment is obtained as  
\begin{align}
f_{V_r}(v_r,\textbf{P},t)&=C(\textbf{P},t)\xi_{v_r}(v_r|\textbf{P},t)\\
&=C(\textbf{P},t)\mathcal{N}_{V_r}(v_r,k_B\mathfrak{T}/m)+g(v_r,\textbf{P},t)
\end{align}
where $g(v_r,\textbf{P},t)$ is an odd function with respect to $v_r$, and 
\begin{align}
&C(\textbf{P}(r,\theta,\phi),t)=\nonumber\\
&C_X(r\sin(\theta)\cos(\phi),t)C_Y(r\sin(\theta)\sin(\phi),t)C_Z(r\cos(\theta),t)
\end{align}
The remainder of the proof, reduces to the one dimensional membrane boundary condition considered in \cite{Koszto2001} but with $v_r$ replaced by $v_x$. Denoting the concentrations at radius $r\leq r_m$ and $r>r_m$, by $C_1(\textbf{P},t)$ and $C_2(\textbf{P},t)$, respectively, the boundary conditions 
\begin{align}
&D_1\frac{\partial C_1(\textbf{P},t)}{\partial r}=D_2\frac{\partial C_2(\textbf{P},t)}{\partial r},\quad r=r_m,\\
&D_1\frac{\partial C_1(\textbf{P},t)}{\partial r}=\frac{z}{1-z}\sqrt {\frac{k_B \mathfrak{T}}{2\pi m}}(C_2(\textbf{P},t)-C_1(\textbf{P},t)), r=r_m,
\end{align}
are obtained. Because of spatial symmetry, we represent the concentration at a point can be full characterized by specifying coordinate $r$. Hence, we have:
\begin{align}
&D_1\frac{\partial C_1(r,t)}{\partial r}=D_2\frac{\partial C_2(r,t)}{\partial r},\quad r=r_m,\\
&D_1\frac{\partial C_1(r,t)}{\partial r}=\frac{z}{1-z}\sqrt {\frac{k_B \mathfrak{T}}{2\pi m}}(C_2(r,t)-C_1(r,t)), r=r_m.
\end{align}
This completes the proof.  
\end{appendix}

\vspace{-0.5cm}
\begin{thebibliography}{21}
\bibitem{Akyl2011}
I. F. Akyildiz, J. M. Jornet, and M. Pierobon, ``Nanonetworks: A new frontier in communications,"  ‎\textit{Communications of the ACM}‎, vol. 54, no. 11, pp. 84-89, Nov. 2011.


\bibitem{NEH13}
 T.~Nakano, A.~Eckford, and T.~Haraguchi, \textit{Molecular communication}, Cambridge University Press, 2013.
 
\bibitem{Farsad2014}
N.~Farsad, B.~Yilmaz, A.~Eckford, C. B.~Chae and W.~Guo, ``A comprehensive survey of recent advancements in molecular communication", \textit{Accepted for publication in the IEEE Communication Surveys and Tutorials}, 2014.

\bibitem{R1}
A.~Noel, K.~Cheung, and R.~Schober, ``Improving receiver performance of diffusive molecular communication with enzymes", \textit{IEEE Transactions on Nanobioscience}, vol.~13, pp.~31-43, Mar. 2014.

\bibitem{MMM14}
M.~Mahfuz, D.~Makrakis, and H.~Mouftah, ``A comprehensive study of sampling-based optimum signal detection in concentration-encoded molecular communication", \textit{IEEE Transaction on NanoBioscience}, vol. 13, no. 3, pp. 208-222, Sep. 2014.
\bibitem{Nakano2013}
T. Nakano, Y. Okaie, and A. V. Vasilakos, ``Transmission rate control for molecular communication among biological nanomachines." \textit{IEEE Journal on Selected Areas in Communications}, vol. 31, no. 12, pp. 835-846, Dec. 2013.



\bibitem{Garralda11a}
 N. Garralda, I.~Llatser, A.~Caballos-Aparicio, and M.~Pierobon, ``Simulation-based evaluation of the diffusion-based physical channel in molecular nanonetworks", \textit{Proc.~IEEE INFOCOM - MoNaCom Workshop}, pp. 443-448, Apr. 2011.


 
 \bibitem{Alberts14}
B. Alberts et al., \textit{Essential cell biology} , 4th edition, Garland Science, 2014.
 
\bibitem{PA10}
 M.~Pierobon and I.~Akyildiz, ``A physical end-to-end model for molecular communications in nanonetworks", \textit{IEEE Journal on Selected Areas in Communications}, pp. 602-611, May 2011.
 
 \bibitem{Arjmandi2013}
H. Arjmandi, A. Gohari, M. Nasiri-Kenari, and Farshid Bateni, ``Diffusion based nanonetworking: A new modulation technique and performance analysis," ‎\textit{IEEE Communications Letters}‎, vol. 17, no. 4, pp. 645-648, Mar. 2013.

\bibitem{Chou2015}
C.T. Chou, ``A Markovian approach to the optimal demodulation of diffusion-based molecular communication networks," \textit{IEEE Transactions on Communications}, vol. 63, no,10, pp. 3728-3743, Oct. 2015.


\bibitem{Olsson2015}
S. B. Olsson, et al., ``Biosynthetic infochemical communication," \textit{Bioinspiration and Biomimetics}, vol. 10, no. 4, Jul. 2015.

\bibitem{Blass2015}
B. Blass, \textit{Basic principles of drug discovery and development}, Elsevier, 2015.


\bibitem{Destex1994}
A. Destexhe, Z. F. Mainen, and T. J. Sejnowski, ``Synthesis of models for excitable membranes, synaptic transmission and neuromodulation using a common kinetic formalism." \textit{Journal of Computational Neuroscience}, vol. 1, no. 3, pp. 195-230, Aug. 1994.





 

\bibitem{Hod1952}
A. L. Hodgkin, A. F. Huxley. ``A quantitative description of membrane current and its application to conduction and excitation in nerve." \textit{The Journal of Physiology}, vol. 117, no. 4, pp.500-544, Aug. 1952.
%



\bibitem{Naqvi1982}
K. Razvi Naqvi, K. J. Mork, and S. Waldenstrøm. ``Reduction of the Fokker-Planck equation with an absorbing or reflecting boundary to the diffusion equation and the radiation boundary condition." \textit{Physical Review Letters}, vol.49, no. 5, Aug. 1982.

\bibitem{Chatur1983}
S. Chaturvedi and G. S. Agarwal. ``A general method for deriving boundary conditions associated with reduced distribution functions," \textit{Zeitschrift fur Physik B Condensed Matter}, vol. 52, no. 3, pp. 247-252, Sep. 1983.

\bibitem{Koszto2001}
T. Koszto lowicz and St. Mr´owczy´nski, ``Membrane boundary conditions", \textit{Acta Physica Polonica. Series B}, vol. 13, no. 1, pp. 217-226, Jan. 2001.


\bibitem{Carslaw1959}
H.S. Carslaw and J. C. Jaeger, \textit{Conduction of heat in solids},  2nd ed., Oxford: Clarendon Press, 1959.

\bibitem{Crank1955}
J. Crank, \textit{The mathematics of diffusion}, 2nd ed., Oxford: University Press, 1955.

\bibitem{Mild1972}
K.H. Mild, ``Diffusion exchange between a membrane bounded sphere and its surrounding, "\textit{The Bulletin of Mathematical Biophysics}, vol. 34, no. 1, pp 93-102, Mar. 1972.

%
\bibitem{Feller1968}
W. Feller, \textit{An introduction to probability theory and its applications}, third edition, New York: John Wiley, 1968.
%
%
%
%


\end{thebibliography}
\end{document}